\documentclass{mecbic}
\providecommand{\anno}{2011}
\providecommand{\firstpage}{41}
\providecommand{\lastpage}{56}
\usepackage{breakurl}        

\usepackage{graphicx}
\usepackage{color}
\usepackage{epsfig}
\usepackage{amssymb}
\usepackage{amsmath}
\usepackage{amsthm}
\usepackage[ansinew]{inputenc}
\usepackage{dsfont}
\usepackage{fancybox,tikz}

\newtheorem{theorem}{Theorem}
\newtheorem{lemma}{Lemma}
\newtheorem{proposition}{Proposition}
\newtheorem{definition}{Definition}

\newtheorem{example}{Example}

\newcommand{\Inference}[2]
       {\mkern-2mu\displaystyle\frac{#1}{\vphantom{,}#2}\mkern-2mu}

\newcommand{\inter}[1]{\stackrel{\scriptsize #1}{\longrightarrow}}

\newcommand{\FORMb}[2]
{\mbox{$ {\cal I},{\cal C},{\cal R}\ \models^{#1}{#2}$}}

\title{Brane Calculi Systems:\\ A Static Preview of their Possible Behaviour}
\author{Chiara Bodei
\institute{Dipartimento di Informatica, Universit\`a di Pisa}
\email{chiara@di.unipi.it} \and Linda Brodo \institute{Dipartimento di Scienze dei Linguaggi, Universit\`a di Sassari} \email{brodo@uniss.it}}
\def\titlerunning{Brane Calculi Systems: a Static Preview}
\def\authorrunning{Chiara Bodei and Linda Brodo}
\begin{document}
\maketitle
\pagestyle{plain}
\pagenumbering{arabic}
\setcounter{page}{41}

\begin{abstract}
We improve the precision of a previous Control Flow Analysis for Brane Calculi \cite{BBC09}, by adding information on the context and introducing causality information on the membranes. This allows us to prove some biological properties on the behaviour of systems specified in Brane Calculi.

\end{abstract}

\section{Introduction}
\label{section:introduction}

In \cite{C05} Cardelli introduced
a family of process calculi, called {\it  Brane Calculi}, endowed with dynamically nested membranes, {focussing on the interactions that happen {\it  on} 
membranes rather than inside them.}
Brane calculi offer a suitable and formal setting for investigating the behaviour of the specified systems, in order to establish the biological properties of interest.
Nevertheless, since the behaviour of a system is usually given in terms of its transition system,
whose size can be huge, especially when modelling complex biological systems, its exploration can be computationally hard. 
One possible solution consists in resorting to static techniques to extract
information on the dynamic behaviour and to check the related dynamic properties,
without actually running the corresponding program. The price is a loss in precision,
because these techniques can only provide approximations of the behaviour. 
However, we can exploit static results to perform a sort of preliminary and not too much expensive screening of
{\it  in silico} experiments. 
In the tradition \cite{cfaBioAmb} of applying static techniques
to process calculi used in modelling biological
phenomena, we present here a contextual and less approximate extension of the Control Flow Analysis for Brane Calculi introduced in \cite{BBC09}.
Control Flow Analysis (CFA) is a static technique, based on Flow Logic \cite{flowlogic}, that provides a variety of automatic and decidable methods and tools for analysing properties of computing systems.
One of the advantages of the CFA is that the obtained information on the behaviour are quite general. As a consequence,
a single analysis can suffice for verifying a variety of properties: different inspections of the CFA results permit to check different properties, with no need of re-analysing it several times. Only the values of interest tracked for testing change accordingly and the definitions of the static counterparts of the dynamic properties must be provided.
Control Flow Analysis provides indeed a {\it  safe over-approximation} of the {\it  exact} behaviour of a system, in terms of the possible reachable configurations. That is, at least all the valid behaviours are captured. More precisely, all those events that the analysis does not consider as possible will {\it  never} occur. On the other hand, the set of events deemed as possible may, or may not, occur in the actual dynamic evolution of the system.
To this end we have improved the precision of the CFA in \cite{BBC09}, by adding information on the context (along the lines of \cite{newcfaBioAmb}) and introducing causality information on the membranes.
Also, this extra-information allows us to refine the static checking of properties related to the spatial structure of membranes.
Furthermore, we focus on causality, since we believe it plays a key role in the understanding of the behaviour of biological systems,
in our case specified in a process algebra like the Brane one.
In order to investigate the possibilities of our CFA to capture some kinds of causal dependencies arising in the MBD version of 
Brane Calculi, we follow \cite{busi}
and its classification, by applying the analysis to the same key examples.
We observe that the analysis is able to capture some of these dependencies.
This is a small improvement in the direction of giving some causal structure to the usually flat CFA results.
The gain in precision is paid in terms of complexity: the presented analysis is rather expensive from a computational point of view.

The paper gets in the research stream dedicated to the application of 
static techniques and, in particular, Control Flow Analysis to bio-inspired process calculi e.g.,~\cite{cfaBioAmb,B08}.
Similar to ours are the works devoted to the analysis of BioAmbients \cite{Bio_Amb}.
In particular, 
\cite{newcfaBioAmb}, where the authors introduce a contextual CFA and \cite{PNN08} where a pathway analysis is exploited for investigating causal properties.
BioAmbients are analysed using instead Abstract Interpretation in \cite{GL05,GL06,GL10}.
The analysis presented in \cite{GL05} records information on the number of occurrences of objects and therefore is able to capture quantitative and causal aspects, necessary to reason on the temporal and spatial structure of processes.
In \cite{GL06}, in a different context, the behaviour of processes is safely approximated and the properties of a fragment of Computation Tree Logic is preserved.
This makes it possible to address temporal properties and therefore some kinds of causality.
Finally, \cite{GL10} presents a static analysis that computes an abstract transition systems for BioAmbients processes, able to validate temporal properties.
Our choice of the Brane calculi depends on the fact they have resulted to be particularly useful for modelling and reasoning about a large class of biological systems, such as the one of the eukaryotic cells that, differently from the prokaryotes, possess a set of internal membranes.
Among the first formalisms used to investigate biological membranes there are the 
P Systems \cite{Paun00}, introduced by P\u{a}un, which  formalise distributed parallel computations biologically-inspired: a biological system is seen as a complex hierarchical structure of nested membranes inspired by the structure of living cells. 
Finally, besides Brane, there are other calculi of interest for our approach, 
that have been specifically defined for 
modelling biological structures such as compartments and membranes,
e.g.,~an extension \cite{LT07} of $\kappa$-calculus \cite{kappa}, Beta Binders \cite{beta1} and the Calculus of Looping Sequences \cite{BMMT07}.

The rest of the paper is organised as follows. 
In Section \ref{brane}, we present the
MBD version of Brane Calculi. 
We introduce the Control Flow
Analysis in Section \ref{CFA}. 
In Section \ref{spatial}, we exploit our analysis to check some properties related to the hierarchical structure of Brane processes.
In Section \ref{causal}, we discuss on which kind of causal information our CFA can capture.
In Section \ref{viral}, the static treatment of Brane PEP action is added and
the whole analysis is applied to 
a model
of infective cycle of the Semliki Forest Virus.
Section \ref{concl} presents some concluding remarks.
Proofs of theorems and lemmata presented
throughout the paper are collected in Appendix \ref{app-proof}.

\section{An overview on Brane Calculi}
\label{brane}

The Brane Calculi~\cite{C05} are a family of calculi defined to describe the interaction amongst membraned component. Specifically, the membrane interactions are explicitly described by means of a set of membrane-based interaction capabilities. A system consists of nested membranes, as described by the following syntax, 
where $n$ is taken from a countable set $\Lambda$ of names.
$$
\begin{array}{lll}
P,Q::= & 
\diamond\ | \ 
P \circ Q\ | \ !P\ | \
\sigma\langle P \rangle^{\mu}
& \mbox{systems $\Pi$} \\
\sigma,\tau ::= 
& 0 \ | \ \sigma|\tau \ | \ !\sigma \ | \ a.\sigma
& \mbox{membrane processes $\Sigma$} \\
a,b ::= 
& mate_n \ | \ mate_n^{\bot}  \ | \ 
bud_n \ | \ bud_n^{\bot} (\rho) \ | \ 
drip(\rho)
& \mbox{MBD actions $\Xi_{MBD}$}
\end{array}
$$
\noindent
The basic structure of a system consists of (sub-)system composition, represented by the monoidal operator $\circ$ (associative, commutative and with $\diamond$ as neutral element).
Replication $!$ is used to represent the composition of an unbounded number of systems or membrane processes. $\sigma\langle P \rangle^{\mu}$ is a {\it  membrane} with content $P$ and interaction capabilities represented by the process $\sigma$. Note that, following  \cite{BBC09}, we annotate membranes with a unique label $\mu$ so as to distinguish the different syntactic occurrences of a membrane.
Note that these labels have no semantic meaning, but they are useful for our CFA. 
We refer to $\mu \in {\bf M}$ as the identity of the membrane $\sigma \langle P \rangle^{\mu}$, where ${\bf M}$ is the finite set of membrane identities. We assume that each considered system is contained in an ideal outermost membrane, identified by a distinguished element $* \in {\bf M}$.

Membranes exhibit interaction capabilities, like the MBD set of actions that model membrane fusion and splitting. The former is modelled by the {\it  mating} operation, the latter can be rendered both by {\it  budding}, that consists in splitting off exactly one internal membrane, and {\it  dripping}, that consists in splitting off one empty membrane. 
For the sake of simplicity, we focus here on the fragment of the calculus without communication primitives and molecular complexes, and with only the MBD actions. The treatment of the alternative set of PEP actions is analogous and it is postponed to Section \ref{causal}, where it is briefly introduced.

Membrane processes $\sigma$ consist of the empty process $0$,  the parallel composition of two processes, represented by the monoidal $|$ operator with $0$ as neutral element, the replication of a process and of the process that executes an interaction {\it  action} $a$ and then behaves as another process $\sigma$. 
Actions for mating ($mate_n$) and budding ($bud_n$) have the corresponding co-actions 
($mate_n^{\bot} $, $bud_n^{\bot} $ resp.) to synchronise with. Here $n$, which identifies a pair of  
complementary action and co-action that can interact, is taken from a countable set $\Lambda$ of names. The actions $bud_n^{\bot}(\rho)$ and $drip(\rho)$ are equipped with a process $\rho$ 
associated to the membrane that will be created when performing budding and dripping actions.

The semantics of the calculi is given in terms of a transition system defined up to a structural congruence and reduction rules.
The standard \emph{structural congruence} $\equiv$
on systems $\Pi$ and membranes $\Xi$ is the least congruence satisfying the
clauses in Table~\ref{structcong}.
\begin{table}[!t]
{\scriptsize
$$
\begin{array}{|c|}
\hline
\begin{array}{ll}
\!\!\! ({\mathcal S}/_{\equiv}, \circ, \diamond) \hbox{\  is a commutative monoid}
\hbox{\ \ \ \ }\hbox{\ \ \ \ }
\hbox{\ \ \ \ }\hbox{\ \ \ \ }
&
\!\!\!({\mathcal B}/_{\equiv}, |, 0)  \hbox{\ is a commutative monoid}
\\
!\diamond \equiv \diamond
&
!0 \equiv 0
\\
! (P \circ Q) \equiv ! P \circ ! Q
&
! (\sigma | \tau) \equiv ! \sigma | ! \tau 
\\
!! P \equiv ! P
&
!! \sigma \equiv ! \sigma
\\
!P \equiv P \circ\ !P
&
!\sigma \equiv \sigma | !\sigma
\\
\multicolumn{2}{c}{
0\langle \diamond \rangle^{\mu} \equiv \diamond
}
\\
\sigma \equiv \tau \Rightarrow \sigma | \rho  \equiv \tau | \rho
&
P \equiv Q \Rightarrow P \circ R \equiv Q \circ R
\\
\sigma \equiv \tau  \Rightarrow !\sigma \equiv !\tau
&
P \equiv Q \Rightarrow !P \equiv ! Q
\\
\sigma \equiv \tau  \Rightarrow a.\sigma \equiv a.\tau
&
P \equiv Q \wedge \sigma \equiv \tau \Rightarrow 	
\sigma \langle P \rangle^{\mu} \equiv \tau \langle Q \rangle^{\mu}
\end{array}
\\
\hline
\end{array}
$$
}
\caption{Structural Congruence for Brane Calculi.}
\label{structcong}
\end{table}
Reduction rules complete the definition of the interleaving semantics. They consist of the basic reaction rules, valid for all brane calculi (upper part of Table~\ref{opsem}) and by the reaction axioms for the MBD version (lower part of Table~\ref{opsem}). We use the symbol $\rightarrow^*$ for the reflexive and transitive closure of the transition relation 
$\rightarrow$.

\begin{table}[htb]
{
$$\begin{array}{|c|}
\hline
\begin{array}{lll}
(Par) & (Brane) & (Struct)\\
\Inference{ P \rightarrow Q}{ P \circ R \rightarrow Q \circ R}
\hspace*{1cm}
&
\Inference{ P \rightarrow Q}{ \sigma \langle P \rangle^{\mu} \rightarrow \sigma \langle Q \rangle^{\mu}}
\hspace*{1cm}
&
\Inference{P \equiv P'~\wedge~ P'\rightarrow Q'~\wedge~Q' \equiv Q
}{P\rightarrow  Q}
\end{array}
\\
[4ex]
\hline
\begin{array}{ll}
\\[.3 ex]
(Mate) &
\textcolor{blue}{mate_n}.\sigma | \sigma_0 \langle P \rangle^{\mu_{P}} \circ
\textcolor{blue}{mate_n^{\bot}}.\tau | \tau_0 \langle Q \rangle^{\mu_Q} \rightarrow
\sigma | \sigma_0 |\tau | \tau_0 \langle P \circ Q \rangle^{\mu_{PQ}} \\
& \mbox{ where }
\mu_{PQ} = {\bf MI_{mate}}(mate_n,\mu_P, mate^{\bot}_n,\mu_Q,\mu_{gp},\mu_p,\mu), 
\\
&
\mbox{ $\mu$ identifies the closest membrane surrounding $\mu_P$ and $\mu_Q$ in the context $\mu_{gp}\mu_p$}\\
(Bud) &
\textcolor{blue}{bud_n^{\bot}(\rho)}.\tau | \tau_0 \langle 
\textcolor{blue}{bud_n}.\sigma | \sigma_0 \langle P \rangle^{\mu_P} \circ Q
\rangle^{\mu_Q}  \rightarrow
\rho  \langle \sigma | \sigma_0 \langle P \rangle^{\mu_P} \rangle^{\mu_{R}} \circ \tau | \tau_0\langle Q \rangle^{\mu_Q} 
\\
& \mbox{ where }
\mu_{R} = {\bf MI_{bud}}(bud_n,\mu_P,bud^{\bot}_n,\mu_Q,\mu_{gp},\mu_p,\mu),
\\
&
\mbox{$\mu$ identifies the closest membrane surrounding $\mu_Q$ in the context $\mu_{gp}\mu_p$}
\\
(Drip) &
\textcolor{blue}{drip(\rho)}.\sigma | \sigma_0 \langle P \rangle^{\mu_P} \rightarrow
\rho \langle  \rangle^{\mu_{R}}  \circ  \sigma | \sigma_0 \langle P\rangle^{\mu_P} 
\\
& \mbox{ where }
\mu_{R} = {\bf MI_{drip}}(drip(\rho),\mu_P,\mu_{gp},\mu_p,\mu),
\\
&
\mbox{$\mu$ identifies the closest membrane surrounding $\mu_P$ in the context $\mu_{gp}\mu_p$}
\\[1ex]
\end{array}
\\
\hline
\end{array}$$
\caption{Reduction Semantics for Brane Calculi.}
\label{opsem}
}
\end{table}
They are quite self-explanatory and we make only a few observations about the labels treatment. Given a system, the set of its membrane identities is finite. Indeed, the structural congruence rule imposes  that
$!\sigma\langle P\rangle^{\mu} \ \equiv\  \sigma\langle P\rangle^{\mu}\ \circ \ !\sigma\langle P\rangle^{\mu}$, i.e.~no new identity label  $\mu$ is introduced by  recursive calls.
A distinguished membrane identity is needed each time a new membrane is generated as a  consequence of a performed action, e.g. the new membrane obtained by 
the fusion of two membranes after a mate synchronisation.
To determine such labels we exploit the functions  ${\bf MI_{mate}}$, ${\bf MI_{bud}}$, and ${\bf MI_{drip}}$ that return fresh and distinct  membrane identities,
depending on the actions and on their syntactic contexts~\cite{BBC09}.
Recall that the number of needed membrane identities is finite,
as finite are the possible combinations of actions and contexts. Therefore, we choose these functions in such a way that,
given an action and the identities of the membranes on which the action (and the
corresponding co-action, if any) reside, the function includes the membrane identity needed
to identify the membrane obtained by firing that action.

\section{A Contextual CFA for Brane Calculi}
\label{CFA}
   
We present an extension of the Control Flow Analysis (CFA), introduced in \cite{BBC09} 
for analysing system specified in Brane Calculi. 
The analysis over-approximates all the possible behaviour of a top-level system $P_{*}$.
In particular, the analysis keeps track of the possible contents of each membrane, thus
taking care of the possible modifications of the containment hierarchy due to the dynamics.
The new analysis, following \cite{newcfaBioAmb}, incorporates context in the style of 2CFA, thus increasing the precision of the approximations w.r.t.~\cite{BBC09}.
Furthermore, the analysis exploits some causality information to further reduce the degree of approximation.
A localised approximation of the contents of a membrane or {\it  estimate} ${\cal I}$ is defined as follows: 
$$
{\cal I} \subseteq {\bf M} \times {\bf M} \times {\bf M} \times ({\bf M} \cup \Xi_{MBD})
$$
\noindent 
Here, $\mu_s \in {\cal I}(\mu_{gp},\mu_p,\mu)$ (that is $(\mu_{gp},\mu_p,\mu,\mu_s) \in {\cal I}$) means that the membrane identified by $\mu$ may surround the membrane identified by $\mu_s$, 
whenever $\mu$ is surrounded by $\mu_p$ and $\mu_p$ is surrounded by $\mu_{gp}$.
The outermost membranes $\mu_{gp}\mu_p$ represent what is called the {\it  context} and that amounts to $**$ when the analysed membrane is at top-level.
Moreover, $a \in {\cal I}(\mu_{gp},\mu_p,\mu)$ means that the action $a$ may reside on and affect the membrane identified by $\mu$, in the context $\mu_{gp}\mu_p$.
Furthermore, the analysis collects two types of some causality information:
\begin{itemize}
\item An approximation of the possible causal circumstances in which a membrane can arise:
$$ 
{\cal C} \subseteq (\Xi_{MBD} \times {\bf M} \times \Xi_{MBD} \times {\bf M} \times {\bf M} \times {\bf M} \times {\bf M}) \uplus 
(\Xi_{MBD} \times {\bf M} \times {\bf M} \times {\bf M} \times {\bf M})
$$
\noindent 
Here  
$(a_n,\mu_P,a_n^{\bot},\mu_Q,\mu_{gp},\mu_p,\mu) \in {\cal C}(\mu_c)$ means that the 
membrane $\mu_c$ can be {\it  causally derived} by the
firing of the action $a_n$ in $\mu_P$ and the coaction $a_n^{\bot}$ in $\mu_Q$, in the
context $\mu_{gp}\mu_p,\mu$.
Similarly, $(a_n,\mu_P,\mu_{gp},\mu_p,\mu) \in {\cal C}(\mu_c)$ for an action $a_n$ like $drip$, without
a co-action.
\item
An approximation of the possible membrane incompatibilities:
$$
{\cal R} \subseteq ({\bf M} \times {\bf M} \times {\bf M}) \times ({\bf M} \times {\bf M} \times {\bf M}) 
$$
\noindent
Here, $((\mu_{gp},\mu_p,\mu), (\mu'_{gp},\mu'_p,\mu)) \in {\cal R}$ means that the membrane $\mu$ in the context $\mu_{gp}\mu_p$ 
cannot interact with the membrane $\mu$ in the context $\mu'_{gp}\mu'_p$, because 
the second membrane is obtained from the first and the first one is dissolved.
\end{itemize}
Note that ${\cal C}$ and ${\cal R}$ are two strict order relations, thus only transitivity property holds.
To validate the correctness of a proposed {estimate} ${\cal I}$, 
we state a set of clauses operating upon judgements 
like 
$\FORMb{\mu_{gp} \mu_p \mu}{P}$.
This judgement
expresses that when the subprocess $P$ of $P_{*}$ is enclosed within a membrane identified 
by $\mu \in {\bf M}$, in the context 
$\mu_{gp}\mu_p \in {\bf M} \times {\bf M} $, then 
${\cal I}$ correctly captures the behaviour of $P$, i.e.~the estimate 
is valid also for all the states $Q$
passed through a computation of $P$.

The analysis is specified in two phases.
First, it checks that ${\cal I}$ describes the initial process. This is done
in the upper part of Table~\ref{analysis}, where the clauses
amount to a syntax-driven structural traversal of process specification.
The clauses rely on the auxiliary function $A$ that collects all the actions in a membrane process 
$\sigma$ and that is reported at the beginning of Table~\ref{analysis}.
Note that the actions collected by $A$, e.g., in $\sigma = \sigma_0.\sigma_1$ are equal to the ones in 
$\sigma' = \sigma_0|\sigma_1$, witnessing
the fact that here the analysis introduces some imprecision and approximation.
he clause for membrane system $\sigma \langle P \rangle^{\mu_s}$ checks that whenever a
membrane ${\mu_s}$ is introduced inside a membrane $\mu$, in the context $\mu_{gp},\mu_p$
the relative hierarchy position must be reflected in
${\cal I}$, i.e.~$\mu_s \in{\cal I}(\mu_{gp},\mu_p,\mu)$.
Furthermore, the actions in $\sigma$ that affect the membrane $\mu_s$ and that are collected in 
$A(\sigma)$, are recorded in ${\cal I}(\mu_p,\mu,\mu_s)$.
Finally, when inspecting the content $P$, the fact that the enclosing membrane is $\mu_s$ in the context $\mu_p\mu$ is recorded, as reflected by the 
judgement $\FORMb{\mu_p \mu \mu_s}{P}$.
The rule for $\diamond$ does not restrict the analysis result, 
while the rules for parallel composition $\circ$, and replication $!$ ensure that the analysis also holds for the immediate sub-systems, by ensuring their traversal.
In particular, note that the analysis of $!P$ is equal to the one of $P$. This is another source of imprecision.
\begin{table}[htb]
\begin{center}
{
\begin{tabular}{|c|}  \hline
\mbox{$\begin{array}{llll}
 \\[.01ex]
{A}(0) = \emptyset  \ \ & \ \
{A}(a.\sigma) = \{a\} \cup {A}(\sigma)
\ \ & \ \
{A}(! \sigma) = {A}(\sigma)
\ \ & \ \
{A}(\sigma_0|\sigma_1) = {A}(\sigma_0) \cup {A}(\sigma_1)\\[1ex]
\end{array} $} 
 \\ \hline
\mbox{$\begin{array}{lll}
 \\[.01ex]
\FORMb{\mu_{gp} \mu_p \mu}{\diamond} & {\sf{iff}} &
{\mathit true}
\\
\FORMb{\mu_{gp} \mu_p \mu}{P \circ Q} & {\sf iff} &
\FORMb{\mu_{gp} \mu_p \mu}{P} \ \wedge \
\FORMb{\mu_{gp} \mu_p \mu}{Q}
\\
\FORMb{\mu_{gp} \mu_p \mu}{! P} & {\sf iff} &
\FORMb{\mu_{gp} \mu_p \mu}{P} 
\\
\FORMb{\mu_{gp} \mu_p \mu}{\sigma \langle P \rangle^{\mu_s}} & {\sf iff} &
\mu_s \in {\cal I}(\mu_{gp},\mu_p, \mu) \ \wedge \ 
A(\sigma) \subseteq {\cal I}(\mu_p, \mu,\mu_s)
\ \wedge \
\FORMb{\mu_p, \mu,\mu_s}{P}
\\[2ex]
\end{array} $} \\ 

\mbox{$\begin{array}{ll}
(Mate) & \ mate_n \in {\cal I}(\mu_p,\mu,\mu_{P})
\wedge
mate_n^{\bot} \in {\cal I}(\mu_p,\mu,\mu_{Q}) \wedge  \mu_{P},\mu_{Q} \in {\cal I}(\mu_{gp},\mu_p,\mu) \ \wedge \\
& ((\mu_p,\mu,\mu_{P}),(\mu_p,\mu,\mu_{Q})) \not\in \mathcal{R}\\
&
\Rightarrow \mu_{PQ}  \in {\cal I}(\mu_{gp},\mu_p,\mu) \ \mbox{ where }
\mu_{PQ} = {\bf MI_{mate}}(mate_n,\mu_P,mate_n^{\bot},\mu_Q,\mu_{gp},\mu_p,\mu) \  \wedge \\
&
{\cal I}(\mu_p,\mu,\mu_{P})  \subseteq  {\cal I}(\mu_p,\mu,\mu_{PQ}) \wedge
{\cal I}(\mu,\mu_{P})  \subseteq  {\cal I}(\mu,\mu_{PQ}) \wedge
{\cal I}(\mu_{P})  \subseteq  {\cal I}(\mu_{PQ})
\\
&
{\cal I}(\mu_p,\mu,\mu_{Q}) \subseteq {\cal I}(\mu_p,\mu,\mu_{PQ}) \wedge
{\cal I}(\mu,\mu_{Q})  \subseteq  {\cal I}(\mu,\mu_{PQ}) \wedge
{\cal I}(\mu_{Q})  \subseteq  {\cal I}(\mu_{PQ})
\\
& \ \wedge \ (mate_n,\mu_P,mate_n^{\bot},\mu_Q,\mu_{gp},\mu_p,\mu) \in {\cal C}(\mu_{PQ}) \ \wedge \\
&
((\mu_p,\mu,\mu_{P}), (\mu_p,\mu,\mu_{PQ})),  ((\mu_p,\mu,\mu_{Q}), (\mu_p,\mu,\mu_{PQ}))
\in \mathcal{R} \ \wedge \\\
&
((\mu,\mu_{P}), (\mu,\mu_{PQ})),  ((\mu,\mu_{Q}), (\mu,\mu_{PQ}))
\in \mathcal{R}
\\[2ex]
(Bud) & \ bud_n \in {\cal I}(\mu, \mu_{Q},\mu_{P})
\wedge
bud_n^{\bot}(\rho) \in {\cal I}(\mu_p, \mu, \mu_{Q}) \wedge \mu_{P} \in {\cal I}(\mu_p,\mu,\mu_{Q}) \wedge \mu_{Q} \in {\cal I}(\mu_{gp},\mu_p,\mu)
\\
&
\Rightarrow
\mu_{R} \in {\cal I}(\mu_{gp},\mu_p,\mu) \mbox{ where }
\mu_{R} = {\bf MI_{bud}}(bud_n,\mu_P,bud_n^{\bot}(\rho), \mu_Q,\mu_{gp},\mu_p,\mu)
\ \wedge \\
& A(\rho) \subseteq {\cal I}(\mu_p,\mu,\mu_R)
 \wedge \mu_P \in {\cal I}(\mu_p,\mu,\mu_R) \ \wedge  \\
 &
 {\cal I}(\mu, \mu_{Q},\mu_{P}) \subseteq {\cal I}(\mu, \mu_{R},\mu_P) \ \wedge \
 {\cal I}(\mu_{Q},\mu_{P}) \subseteq {\cal I}(\mu_{R},\mu_P) \
\\
& \wedge  \ (bud_n,\mu_P,bud_n^{\bot}(\rho),\mu_Q,\mu_{gp},\mu_p,\mu) \in {\cal C}(\mu_{R})\\
& \ \wedge \ ((\mu, \mu_{Q},\mu_{P}), (\mu, \mu_{R},\mu_{P})) \in \mathcal{R} 
\ \wedge \ ((\mu_{Q},\mu_{P}), (\mu_{R},\mu_{P})) \in \mathcal{R}
\\[2ex]
(Drip) & \ drip(\rho) \in {\cal I}(\mu_p, \mu, \mu_{P})
\wedge
\mu_{P} \in {\cal I}(\mu_{gp},\mu_p,\mu) 
\\
&
\Rightarrow 
\mu_{R} \in {\cal I}(\mu_{gp},\mu_p,\mu) 
\mbox{ where }
\mu_{R} = {\bf MI_{drip}}(drip(\rho),\mu_P,\mu_{gp},\mu_p,\mu) \ \wedge \ A(\rho) \subseteq {\cal I}(\mu_p,\mu,\mu_R)  \\
& \wedge \ (drip(\rho),\mu_P,\mu_{gp},\mu_p,\mu) \in {\cal C}(\mu_{R})
\\
[2ex]
\end{array} $} \\ \hline
\end{tabular}
}
\end{center}
\caption{Analysis for Brane Processes}
\label{analysis}
\end{table}

Secondly, the analysis checks that ${\cal I}$ also takes into account the dynamics of the process under consideration; in particular, the dynamics of the containment hierarchy of membranes.
This is expressed by
the closure conditions in the lower part of Table \ref{analysis} that
mimic the semantics, 
by modelling, without exceeding the precision boundaries of the analysis,
the semantic preconditions
and the consequences of the possible actions.
More precisely, each precondition checks whether a pair of complementary actions could possibly enable the firing of a transition according to ${\cal I}$.
The conclusion imposes the additional requirements on ${\cal I}$ that are necessary to give a valid prediction of the analysed action.

Consider e.g., the clause for $(Mate)$ (the other clauses are similar).
If (i) there exists an occurrence of a mate action: $mate_n \in {\cal I}(\mu_p,\mu,\mu_{P})$;
(ii)  there exists an occurrence of the corresponding co-mate action:
$mate^{\bot}_n \in {\cal I}(\mu_p,\mu,\mu_{Q})$;
(iii) the corresponding membranes are siblings: $\mu_{P},\mu_{Q} \in {\cal I}(\mu_{gp},\mu_p,\mu)$,
(iv) the redexes are not incompatible, 
i.e.~the corresponding membranes can interact: $((\mu_p,\mu,\mu_{P}),(\mu_p,\mu,\mu_{Q}))\not\in\mathcal{R}$
then the conclusion of the clause expresses the
effects of performing the transition $(Mate)$.
In this case, we have that ${\cal I}$ must reflect that 
(i) there may exist a membrane $\mu_{PQ}$ inside $\mu$, in the context $\mu_{gp}\mu_p$, at the same nesting level of the membranes $\mu_{P}$ and $\mu_{Q}$; and 
(ii) the contents of $\mu_{P}$ and of $\mu_{Q}$, their children and their grandchildren, may also be included in $\mu_{PQ}$.
Note that the contribution changes depending on whether we consider $\mu_{P}$ ($\mu_{Q}$, resp.), their children or their grandchildren.
With the inclusion ${\cal I}(\mu,\mu_P) \subseteq {\cal I}(\mu,\mu_{PQ})$ we mean that
for each $\mu_s$ in the context
$\mu,\mu_P$, all the elements in  
${\cal I}(\mu,\mu_P,\mu_s)$ are included in ${\cal I}(\mu,\mu_{PQ},\mu_s)$.
Similarly, with ${\cal I}(\mu_P) \subseteq {\cal I}(\mu_{PQ})$ we mean that
for each $\mu_{gs}$ in the context $\mu_P\mu_s$, and in turn for each $\mu_s$ in the context 
$\mu,\mu_P$, all the 
elements in ${\cal I}(\mu_P,\mu_s, \mu_{gs})$ belong to ${\cal I}(\mu_{PQ},\mu_s, \mu_{gs})$.
We use a similar notation for the relation ${\cal R}$.
(iii)
The membrane $\mu_{PQ}$ is the result of the transition $(Mate)$, performed by the two membranes $\mu_{P}$ and $\mu_{Q}$, in the context $\mu_{gp}\mu_p\mu$, 
as witnessed by the corresponding entry in the component ${\cal C}$;
(iv) the new membrane $\mu_{PQ}$ is $incompatible$ with the $\mu_{P}$ and $\mu_{Q}$, because $\mu_{PQ}$, derived by the transition $(Mate)$, follows both
$\mu_{P}$ and $\mu_{Q}$.
Note the similar incompatibility between the membrane $\mu_P$ in the context $\mu \mu_{Q}$ before the $(Bud)$ transition and the derived one
$\mu_P$ in the context $\mu \mu_{R}$.
The above requirements correspond to the application of the semantic rule $(Mate)$ that would result in the fusion of the two membranes.

Note that, since the new membrane $\mu_{PQ}$ inherits the prefix actions
that affected the membranes $\mu_{P}$ and $\mu_{Q}$, it inherits also $mate_n$ and $mate^{\bot}_n$ (we write in red this kind of {\it  imprecise} inclusions).
This is due to over-approximation, even though it is harmless: the two prefix actions cannot be further used to predict a communication because they both occur in 
${\cal I}(\mu_p,\mu,\mu_{PQ})$.
Still, the presence of both $mate_n \in {\cal I}(\mu_p,\mu,\mu_{P})$ and $mate^{\bot}_n \in {\cal I}(\mu_p,\mu,\mu_{PQ})$ could lead to predict another interaction that is impossible at run time.
Thanks to ${\cal R}$, we can safely exclude it, thus gaining precision.
This gain is obtained in general: ${\cal R}$ collects indeed pairs of capabilities that could be syntactically compatible with an interaction, 
but that cannot really interact, because they dynamically occur in membranes that are not simultaneously present.

The gain in precision is paid in terms of complexity: the presented analysis is rather expensive from a computational point of view, due to the introduction of contexts and to the possibly high
number of different membrane names. Both these features may lead to an explosion of the possible reachable configurations.

\begin{example}
To illustrate how our CFA work we use two simple examples. The emphasis is on the process algebraic structures and not on their biological expressiveness.
We first report an application of it to a simple process $P$, illustrated in \cite{BBC09} (and in turn taken from \cite{busi}).
We consider $P$ and the following possible computations, where $\rho_1$ and $\rho_2$ 
are not specified as they are not relevant here.

$
\begin{array}{l}
\! \! \! \! \! \! \! \! \!  P = (mate_n | bud_m^{\bot}(\rho_1)) \langle bud_m \langle \rangle^{\mu_{P_0}} \circ bud_o  \langle \rangle^{\mu_{P_1}}
 \rangle^{\mu_{P}} 
\circ  
(mate_n^{\bot} | bud_o^{\bot}(\rho_2)) \langle \rangle^{\mu_{Q}} 
 \inter{mate_n} \\
P_1 = (bud_m^{\bot}(\rho_1) | bud_o^{\bot}(\rho_2))
\langle bud_m \langle \rangle^{\mu_{P_0}} \circ bud_o \langle \rangle^{\mu_{P_1}}
\circ \diamond \rangle^{\mu_{PQ}} \inter{bud_m}\\
P_2= \rho_1 \langle \langle \rangle^{\mu_{P0}} \rangle^{\mu^{0}_{R1}} \circ 
bud_o^{\bot}(\rho_2))
\langle bud_o \langle \rangle^{\mu_{P_{1}}}
\circ \diamond \rangle^{\mu_{PQ}} \inter{bud_o}\\
P_3 = \rho_1 \langle \langle \rangle^{\mu_{P0}} \rangle^{\mu^{0}_{R1}} \circ 
\rho_2 \langle \langle \rangle^{\mu_{P1}} \rangle^{\mu_{R2}} \circ 
\langle  \diamond \rangle^{\mu_{PQ}} 
\end{array}
\normalsize
$


$
\begin{array}{l}
\! \! \! \! \! \! \! \! \!  P = (mate_n | bud_m^{\bot}(\rho_1)) \langle bud_m \langle \rangle^{\mu_{P_0}} \circ bud_o  \langle \rangle^{\mu_{P_1}}
 \rangle^{\mu_{P}} 
\circ  
(mate_n^{\bot} | bud_o^{\bot}(\rho_2)) \langle \rangle^{\mu_{Q}} 
 \inter{bud_m} \\
P'_1 = \rho_1 \langle \langle \rangle^{\mu_{P0}} \rangle^{\mu^{1}_{R1}} \circ 
mate_n \langle \langle bud_o \rangle^{\mu_{P_1}} \rangle^{\mu_P} \circ (mate_n^{\bot} | bud_o^{\bot}(\rho_2)) \langle \rangle^{\mu_{Q}} 
\inter{mate_m}\\
P'_2 = \rho_1 \langle \langle \rangle^{\mu_{P0}} \rangle^{\mu^{1}_{R1}} \circ 
bud_o^{\bot}(\rho_2) \langle \langle bud_o \rangle^{\mu_{P_1}}  \rangle^{\mu_{PQ}} \inter{bud_o} \\
P'_3 =  \rho_1 \langle \langle \rangle^{\mu_{P0}} \rangle^{\mu^{1}_{R1}} \circ 
\rho_2 \langle \langle \rangle^{\mu_{P1}} \rangle^{\mu_{R2}} \circ 
\langle  \diamond \rangle^{\mu_{PQ}} 
\end{array}
\normalsize
$

\begin{table}
\begin{center}
{
\small
\begin{tabular}{|c|} \hline
\mbox{$\begin{array}{ll}
\mu_{P_0},\mu_{P_1} \in {\cal I}(*,*,\mu_{P}),
\mu_{P},\mu_{Q}\in {\cal I}(*,*,*) &
\mu_{P_0},\mu_{P_1} \in {\cal I}(*,*,\mu_{PQ}), \mu_{PQ} \in {\cal I}(*,*,*)
\\
\mu^{i}_{R1} \in {\cal I}(*,*,*), \mu_{P_0} \in {\cal I}(*,*,\mu^{i}_{R1}) \ (i = 0,1)
&
\mu_{R2} \in {\cal I}(*,*,*), \mu_{P_1} \in {\cal I}(*,*,\mu_{R2})
\\
\begin{array}{lllllll}
mate_n & \in & {\cal I}(*,*,\mu_{P}), & \textcolor{red}{{\cal I}(*,*,\mu_{PQ})},\\
mate_n^{\bot} & \in & {\cal I}(*,*,\mu_{Q}), & \textcolor{red}{{\cal I}(*,*,\mu_{PQ})}, \\
bud_m & \in & {\cal I}(*,\mu_P,\mu_{P_0}), & {\cal I}(*,\mu_{PQ},\mu_{P_0}), & \textcolor{red}{{\cal I}(*,\mu^{i}_{R1},\mu_{P_0})},\\
bud_m^{\bot}(\rho_1) & \in & {\cal I}(*,*,\mu_{P}), & {\cal I}(*,*,\mu_{PQ}), & \\
bud_o & \in & {\cal I}(*, \mu_P, \mu_{P_1}), & {\cal I}(*,\mu_{PQ}, \mu_{P_1}), & \textcolor{red}{{\cal I}(*,\mu_{R2},\mu_{P_1})},\\
bud_o^{\bot}(\rho_2) & \in & {\cal I}(*,*,\mu_{Q}), & {\cal I}(*,*,\mu_{PQ}),   &
\end{array}
\\
((*,*,\mu_{P}), (*,*,\mu_{PQ})) \in {\cal R} 
&
((*,*,\mu_{Q}),(*,*,\mu_{PQ})) \in {\cal R} 
\\
((*,\mu_{P},\mu_{P_{0}}), (*,\mu_{PQ},\mu_{P_{0}})) \in {\cal R} 
& ((*,\mu_{P},\mu_{P_{1}}), (*,\mu_{PQ},\mu_{P_{1}})) \in {\cal R}
\\
((*,\mu_{P},\mu_{P_{0}}),(*,\mu^{i}_{R1},\mu_{P_{0}}) \in {\cal R}
&
((*,\mu_{P},\mu_{P_1}),(*,\mu_{R2},\mu_{P_1})) \in  {\cal R} 
\\
(mate_n,\mu_P,mate_{n}^{\bot},\mu_Q,*,*,*) \in {\cal C}(\mu_{PQ})
&
(bud_0,\mu_{P_1},bud^{\bot}_o,\mu_{PQ},*,*,*) \in {\cal C}(\mu_{R2})\\

(bud_m,\mu_{P_0},bud^{\bot}_m,\mu_{PQ},*,*,*)\in {\cal C}(\mu^{i}_{R1}) &
\end{array} $} \\ \hline
\end{tabular}
}
\end{center}
\caption{Some entries of the Example 1 Analysis}
\label{ex}
\end{table}
\noindent
The main entries of the analysis are reported in Table \ref{ex}, 
where $**$ identifies the ideal outermost context in which the system top-level membranes are.
We write in red the entries due to approximations, but not reflecting the dynamics.
Furthermore, we pair the inclusions of actions and of the corresponding co-actions, in order to emphasise which are the pairs of prefixes that lead to the prediction
of a possible communication.
It is easy to check that ${\cal I}$ is a valid estimate by following the two stage procedure explained above.

To understand in which way the ${\cal R}$ component refines the analysis, note that
since the analysis entries include
$\textcolor{red}{mate_n \in {\cal I}(*,*,\mu_{PQ})}$ and $mate_n^{\bot}\in {\cal I}(*,*,\mu_{Q})$, without the check on the ${\cal R}$ component,
we can predict a transition between the two membranes $\mu_{PQ}$ and $\mu_{P}$. This transition is not possible instead, because
$\mu_{PQ}$ is causally derived by $\mu_{P}$.

Note that although the CFA offers in general an over-approximation of the possible dynamic behaviour, in this example the result is rather precise.
The transition $mate_n$ is predicted as possible, since 
its precondition requirements are satisfied.
Indeed, we have that $mate_n \in {\cal I}(*,*,\mu_{P})$ 
$mate_n^{\bot} \in {\cal I}(*,*,\mu_{Q})$,  and $\mu_P$ and $\mu_Q$ are sibling 
and $causally \ compatible$ membranes. 
Also the transition on $bud_m$ is initially possible and this result is actually predicted by the analysis, since $bud_{m}\in {\cal I}(*, \mu_P,\mu_{P_{0}})$ and 
$bud_m^{\bot} \in {\cal I}(*, *,\mu_P)$, 
with $\mu_{P_{0}} \in {\cal I}(*, *,\mu_P)$, i.e. $P$ is the father of $P_0$.
Instead, we can observe that the transition on $bud_o$ cannot be performed in the initial system.
Indeed,  $bud_o$ resides on the membrane $\mu_{P_{1}}$ in the context $*\mu_P$, while
the coaction  $bud_o^{\bot}$ resides on $\mu_Q$ that is not the father of $\mu_{P_{1}}$.
The transition on $bud_o$ can be performed instead in the membrane $\mu_{P_{1}}$ in the context $*\mu_{PQ}$, that is the membrane introduced by the previous $mate_n$ transition.

\end{example}

\begin{example}
We now apply our CFA to another process $P$, taken from \cite{busi}.
We consider $P$ and the following possible computations.

$
\begin{array}{l}
\! \! \! \! \! \! \! \! \!  P = mate_n \langle 
(mate_m | mate_o) \langle \rangle^{\mu_{P_0}} \circ mate_o^{\bot}  \langle  \rangle^{\mu_{P_1}} 
 \rangle^{\mu_{P}} 
\circ  
mate_n^{\bot}. \langle mate_m^{\bot} \langle \rangle^{\mu_{Q_0}} \rangle^{\mu_{Q}} 
 \inter{mate_n} \\
P_{1} = \langle 
(mate_m | mate_o) \langle \rangle^{\mu_{P_0}} \circ mate_o^{\bot}  \langle  \rangle^{\mu_{P_1}} 
\circ  
 mate_m^{\bot} \langle \rangle^{\mu_{Q_0}} \rangle^{\mu_{PQ}}  
  \inter{mate_m} \\
P_{2} = \langle 
mate_o \langle \rangle^{\mu_{P_0Q_0}} \circ mate_o^{\bot}  \langle  \rangle^{\mu_{P_1}} 
\rangle^{\mu_{PQ}}  
  \inter{mate_0} 
P_{3} = \langle 
\langle \rangle^{\mu_{P_0Q_0P_1}}
\rangle^{\mu_{PQ}}  
\end{array}
\normalsize
$

$
\begin{array}{l}
P_{1} = \langle 
(mate_m | mate_o) \langle \rangle^{\mu_{P_0}} \circ mate_o^{\bot}  \langle  \rangle^{\mu_{P_1}} 
\circ  
mate_m^{\bot} \langle \rangle^{\mu_{Q_0}} \rangle^{\mu_{PQ}}  
  \inter{mate_o} \\
P'_{2} = \langle 
mate_m \langle \rangle^{\mu_{P_0P_1}}\circ  
 \langle mate_m^{\bot} \langle \rangle^{\mu_{Q_0}}
\rangle^{\mu_{PQ}}  
  \inter{mate_m} 
P'_{3}= \langle 
\langle \rangle^{\mu_{P_0P_1Q_0}}
\rangle^{\mu_{PQ}}  
\end{array}
\normalsize
$

$
\begin{array}{l}
\! \! \! \! \! \! \! \! \!  P = mate_n \langle 
(mate_m | mate_o) \langle \rangle^{\mu_{P_0}} \circ mate_o^{\bot}  \langle  \rangle^{\mu_{P_1}} 
 \rangle^{\mu_{P}} 
\circ  
mate_n^{\bot}. \langle mate_m^{\bot} \langle \rangle^{\mu_{Q_0}} \rangle^{\mu_{Q}} 
 \inter{mate_o} \\
P''_{1} = mate_n \langle 
mate_m  \langle \rangle^{\mu'_{P_0P_1}} 
 \rangle^{\mu_{P}} 
\circ  
mate_n^{\bot}. \langle mate_m^{\bot} \langle \rangle^{\mu_{Q_0}} \rangle^{\mu_{Q}}   
\inter{mate_n} \\
P''_{2} =  \langle 
mate_m  \langle \rangle^{\mu'_{P_0P_1}} 
\circ  
mate_m^{\bot} \langle \rangle^{\mu_{Q_0}} \rangle^{\mu_{PQ}}   
  \inter{mate_m} 
P''_{3} = \langle 
\langle \rangle^{\mu'_{P_0P_1Q_0}} 
 \rangle^{\mu_{PQ}}     
\end{array}
\normalsize
$

\begin{table}
\begin{center}
{
\small
\begin{tabular}{|c|} \hline
\mbox{$\begin{array}{l}
\mu_P,\mu_Q,\mu_{PQ} \in  {\cal I}(*,*,*),\\
\mu_{P_0},\mu_{P_1} \in {\cal I}(*,*,\mu_{P}),{\cal I}(*,*,\mu_{PQ}),
\mu_{Q_0} \in {\cal I}(*,*,\mu_{Q}),{\cal I}(*,*,\mu_{PQ})
\\
\mu_{P_{0}Q_{0}}  \in {\cal I}(*,*,\mu_{PQ}),
\mu_{P_{0}Q_{0}P_{1}} \in {\cal I}(*,*,\mu_{PQ}),\\
\mu_{P_{0}P_{1}} \in {\cal I}(*,*,\mu_{PQ}), \mu'_{P_{0}P_{1}} \in {\cal I}(*,*,\mu_{P}), {\cal I}(*,*,\mu_{PQ}),\\
\mu_{P_{0}P_{1}Q_{0}} \in {\cal I}(*,*,\mu_{PQ}),
\mu'_{P_{0}P_{1}Q_{0}} \in {\cal I}(*,*,\mu_{PQ}),\\
\begin{array}{lllllll}
mate_n &\in& {\cal I}(*,*,\mu_{P}), \\ 
mate_n^{\bot} &\in& {\cal I}(*,*,\mu_{Q}), \\
mate_m &\in& {\cal I}(*,\mu_{P},\mu_{P_0}), & {\cal I}(*,\mu_{PQ},\mu_{P_0}),
& {\cal I}(*,\mu_{PQ},\mu_{P_0P_1}),\\
mate_m^{\bot} &\in& {\cal I}(*,\mu_{Q}, \mu_{Q_0}), &{\cal I}(*,\mu_{PQ}, \mu_{Q_0}), 
& {\cal I}(*,\mu_{PQ}, \mu_{Q_0}), \\
mate_m &\in& 
{\cal I}(*,\mu_{P},\mu'_{P_0P_1}), & {\cal I}(*,\mu_{PQ},\mu'_{P_0P_1}), &   \\ 
mate_m^{\bot} &\in& 
{\cal I}(*,\mu_{Q}, \mu_{Q_0}), & {\cal I}(*,\mu_{PQ}, \mu_{Q_0}), \\
mate_o &\in& {\cal I}(*,\mu_{P}, \mu_{P_0}), 
& {\cal I}(*,\mu_{PQ}, \mu_{P_0Q_0}), & {\cal I}(*,\mu_{PQ}, \mu_{P_0}),\\
mate_o^{\bot} &\in& {\cal I}(*,\mu_{P}, \mu_{P_1}), & {\cal I}(*,\mu_{PQ}, \mu_{P_1}) &  {\cal I}(*,\mu_{PQ}, \mu_{P_1}),
\\
\end{array}
\\
((*,*,\mu_{P}), (*,*,\mu_{PQ})) \in {\cal R} 
\\
((*,*,\mu_{Q}), (*,*,\mu_{PQ})) \in {\cal R} 
\\
((*,\mu_{P},\mu_{P_{i}}), (*,\mu_{PQ},\mu_{P_{i}})) \in {\cal R} 
\ i = 0,1
\\
(mate_n,\mu_P,mate_{n}^{\bot},\mu_Q,*,*,*) \in {\cal C}(\mu_{PQ})
\\
(mate_m, \mu_{P_0},mate^{\bot}_n,\mu_{Q_0},*,*,\mu_{PQ}) \in {\cal C}(\mu_{P_0Q_0})\\
(mate_m, \mu_{P_0P_1},mate^{\bot}_m,\mu_{Q_0},*,*,\mu_{PQ}) \in {\cal C}(\mu_{P_0P_1Q_0}) \\
(mate_o, \mu_{P_0Q_0},mate^{\bot}_o,\mu_{P_1},*,*,\mu_{PQ}) \in {\cal C}(\mu_{P_0Q_0P_1}) \\
(mate_o, \mu_{P_0},mate^{\bot}_o,\mu_{P_1},*,*,\mu_{PQ}) \in {\cal C}(\mu_{P_0P_1}) \\
(mate_o, \mu_{P_0},mate^{\bot}_o,\mu_{P_1},*,*,\mu_{P}) \in {\cal C}(\mu'_{P_0P_1})
\end{array} $} \\ \hline
\end{tabular}
}
\end{center}
\caption{Some entries of the Example 2 Analysis}
\label{ex2}
\end{table}
\noindent
The main entries of the analysis are reported in Table \ref{ex2}, where 
we do not include the entries due to approximations, but not reflecting the dynamics.
As before, we pair the inclusions of actions and of the corresponding co-actions, in order to emphasise which are the pairs of prefixes that lead to the prediction
of a possible communication. This motivates some redundancies in the entries.
Also in this example, the CFA result is rather precise.


\end{example}

\paragraph{\bf Semantic Correctness}

Our analysis is semantically correct with respect to the given semantics, i.e.~a valid estimate enjoys the following subject reduction property with respect to the semantics.
\begin{theorem}\textbf{(Subject Reduction)}\label{sbj-red}\\
If $P \rightarrow Q$ and $\FORMb{\mu_{gp} \mu_p \mu}{P}$ then also 
$\FORMb{\mu_{gp} \mu_p \mu}{Q}$.
\end{theorem}

This result depends on the fact that
analysis is invariant under the structural congruence, as stated below.

\begin{lemma}\textbf{(Invariance of Structural Congruence)}\label{congr}
If $P \equiv Q$ and we have that  $\FORMb{\mu_{gp} \mu_p \mu}{P}$  then also 
$\FORMb{\mu_{gp} \mu_p \mu}{Q}$.
\end{lemma}

Moreover, it is possible to prove that there always exists a least estimate (see \cite{BBC09} for a
similar statement and proof).

\section{CFA for Spatial Structure Properties}
\label{spatial}

Control Flow Analysis provides indeed a {\it  safe over-approximation} of the {\it  exact} behaviour of a system, that is, at least all the valid behaviours are captured. More precisely, all those events that the analysis does not consider as possible will {\it  never} occur. On the other hand, the set of events deemed as possible may, or may not, occur in the actual dynamic evolution of the system.
The 2CFA gains precision w.r.t~the 0CFA presented in \cite{BBC09} and the incompatibility relation ${\cal R}$ increases this gain.
 In the next section, we will discuss on the contribution of the component ${\cal C}$.

We can exploit our analysis to check {\it  spatial structure} properties, of the membranes included in the system under consideration.
In particular, because of over-approximation, we can ask negative questions like whether:
(i) a certain interaction capability $c$ {\it  never affects} the membrane labelled $\mu$, i.e.~it never occurs in the 
membrane process of the membrane labelled $\mu$;
(ii) the membrane labelled $\mu$ {\it  never ends up} in the membrane labelled $\mu'$.

Suppose we have all the possible labels $\mu$ of the possible membranes arising at run time. Then we can precisely define
the above informally introduced properties.
We first give the definition of the dynamic property, then the corresponding static property and, finally,
we show that the static property implies the dynamic one.
For each static property, we check for a particular content in the component ${\cal I}$. 

\begin{definition}[Dynamic: c never on $\mu$]
Given a process $P$ including a membrane labelled $\mu$, we say that the capability $c$ {\it  never affects} the membrane labelled $\mu$
if there not exists a derivative $Q$ such that $P \rightarrow^* Q$, in which the capability $c$ can affect the membrane labelled $\mu$.
\end{definition}

\begin{definition}[Static: c never on $\mu$]
Given a process $P$ including a membrane labelled $\mu$, we say that the capability $c$ {\it  never appears on} the membrane 
labelled $\mu$ if and only if there exists an estimate $({\cal I}, {\cal R},{\cal C})$ such that:
$c \not\in {\cal I}(\mu_{gp},\mu_p,\mu)$ for each possible context $\mu_{gp}\mu_p$.
\end{definition}

\begin{theorem}
Given a process $P$ including a membrane labelled $\mu$, then
if $c$ never appears on the membrane 
labelled $\mu$, then 
the capability $c$ never affects the membrane labelled $\mu$.
\end{theorem}

\begin{definition}[Dynamic: $\mu'$ never inside $\mu$]
Given a process $P$ including a membrane labelled $\mu$ and a membrane labelled $\mu'$,
we say that the membrane $\mu'$ {\it  never ends up inside} the membrane labelled $\mu$
if there not exists a derivative $Q$ such that $P \rightarrow^* Q$, in which $\mu'$ occurs inside the membrane $\mu$.
\end{definition}

\begin{definition}[Static: $\mu'$ never inside $\mu$]
Given a process $P$ including a membrane labelled $\mu$, we say that $\mu'$ {\it  never appears inside} the membrane 
labelled $\mu$ if and only if there exists an estimate $({\cal I}, {\cal R},{\cal C})$ such that:
$\mu' \not\in {\cal I}(\mu_{gp},\mu_p,\mu)$ for each possible context $\mu_{gp}\mu_p$.
\end{definition}

\begin{theorem}
Given a process $P$ including a membrane labelled $\mu$ and a membrane labelled $\mu'$, then
if  $\mu'$  never appears inside the membrane 
labelled $\mu$, then 
the membrane $\mu'$  never ends up inside the membrane labelled $\mu$.
\end{theorem}

Back to our first running example, we can prove, for instance, that the capability
$bud^{\bot}_o$ never affects the membrane labelled $\mu_P$.
This can be checked by looking in the CFA entries, for the content of ${\cal I}(*,*,\mu_P)$,
that indeed does not include $bud^{\bot}_o$.
Intuitively, this explains the fact that the $bud_o$ synchronisation is not syntactically possible in the context $\mu_P$, whose
sub-membrane $\mu_{P_0}$ is affected by $bud_o$.

In our second running example, we can prove instead, for instance, that the membrane 
$\mu_{P_0Q_0}$ 
never ends up inside the membrane labelled $\mu_{P}$, where
$\mu_{P_0Q_0} = {\bf MI_{mate}}(mate_m, \mu_{P_0},mate^{\bot}_n,\mu_{Q_0},*,*,\mu_{PQ})$. 
Indeed, by inspecting the CFA results, we have that
$\mu_{P_0Q_0} \not\in {\cal I}(*,*,\mu_P)$.
Intuitively, this corresponds to the fact that the $mate_m$ synchronisation is not syntactically possible in the context $\mu_P$, 
because $\mu_{P_0}$ and $\mu_{Q_0}$ are not siblings,
while it is in the context $\mu_{PQ}$.

Similarly, we can mix ingredients and introduce new properties, e.g.~one can ask whether
two membranes labelled $\mu'$ and $\mu''$, {\it  never end up (occur) together} in the same membrane $\mu$.
On the static side, this amounts to checking whether there exists an estimate $({\cal I}, {\cal R},{\cal C})$ such that:
for all possible context $\mu_{gp}\mu_p$, $\mu' \in {\cal I}(\mu_{gp},\mu_p,\mu) \wedge \mu'' \not\in {\cal I}(\mu_{gp},\mu_p,\mu)$ 
or $\mu' \not\in {\cal I}(\mu_{gp},\mu_p,\mu) \wedge \mu'' \in {\cal I}(\mu_{gp},\mu_p,\mu)$. 
Note that a single analysis can suffice for verifying all the above properties: 
only the values of interest tracked for testing change.


\section{Discussion on Causal Information}
\label{causal}

Understanding the causal relationships between the actions performed by a process is a relevant issue for all process algebras used in Systems Biology. 
Although our CFA approximates the possible reachable configurations, we are able to extract some information on the causal relations among these configurations.
To investigate these possibilities of our CFA, we follow \cite{busi}
where different kinds of causal dependencies are described and classified, 
by applying our analysis to the same key examples.

The first kinds are called {\it  structural causality} and {\it  synchronisation causality} and are typical of all process algebras.
Structural causality arises from the prefix structure of terms, as in 

$
\begin{array}{l}
P = drip(\sigma).drip(\rho) \langle \rangle^{\mu_{P}} \inter{drip}
\sigma \langle \rangle^{\mu_R} \circ drip(\rho) \langle \rangle^{\mu_{P}}
\inter{drip} 
\sigma \langle \rangle^{\mu_R} \circ \rho \langle \rangle^{\mu'_R}
\circ \langle \rangle^{\mu_{P}}
\end{array}
$

\noindent
where the action on $drip(\rho)$ depends on the one on $drip(\sigma)$, since the second action is not reachable until the first has fired.
Synchronisation causality arises when an action depends on a previous synchronisation as in:

$
\begin{array}{l}
P = drip(\sigma_1).mate_n.drip(\tau_1) \langle \rangle^{\mu_{P'_0}} \circ
drip(\sigma_2).mate_n.drip(\tau_2) \langle \rangle^{\mu_{P'_1}} \inter{drip}^2 \\
P' = \sigma_1 \langle \rangle^{\mu_{R1}} \circ \sigma_2 \langle \rangle^{\mu_{R2}} \circ
mate_n.drip(\tau_1) \langle \rangle^{\mu_{P'_0}} \circ
mate_n.drip(\tau_2) \langle \rangle^{\mu_{P'_1}} \inter{mate_n} \\
P'' = \sigma_1 \langle \rangle^{\mu_{R1}} \circ \sigma_2 \langle \rangle^{\mu_{R2}} \circ 
\langle \rangle^{\mu_{P'_0P'_1}} 
\end{array}
$

\noindent
where, the mate action is possible only when both
$drip(\sigma_1)$ and $drip(\sigma_2)$ have been performed, and the following $drip(\tau_1)$ and $drip(\tau_2)$ depend on the previous
mate synchronisation.
Our CFA is not able to capture these kinds of dependencies, because of the $A()$ function definition, according to which $A(\sigma.\tau) = A(\sigma) \cup A(\tau)$.  
In other words, the CFA disperses the order between prefixes.

According to \cite{busi}, when an action is performed on a membrane it impacts only on its continuation and not on the whole process on the membrane, e.g., in:

$
\begin{array}{l}
P = (mate_n | drip(\sigma)) \langle \rangle^{\mu_{P}} \circ mate_n^{\bot} \langle \rangle^{\mu_Q}
 \inter{mate_n}
drip(\sigma) \langle \rangle^{\mu_{PQ}} \inter{drip}
\sigma \langle \rangle^{\mu_R} \circ \langle \rangle^{\mu_{PQ}} 
\end{array}
$

$
\begin{array}{l}
P = (mate_n | drip(\sigma)) \langle \rangle^{\mu_{P}} \circ mate_n^{\bot} \langle \rangle^{\mu_Q}
 \inter{drip}
\sigma \langle \rangle^{\mu'_R} \circ (mate_n) \langle \rangle^{\mu_{P}} \circ mate_n^{\bot} \langle \rangle^{\mu_Q}
\\ \inter{mate_n}
\sigma \langle \rangle^{\mu'_R} \circ \langle \rangle^{\mu_{PQ}} 
\end{array}
$

\noindent
the drip operation can be considered causally independent form the mate operation, 
because it can be executed regardless of the fact that the mate interaction has been performed.

Our analysis reflects this, because we have
${\cal C}(\mu_R) \ni 
(drip(\sigma),\mu_P,*,*,*)$ and also that
${\cal C}(\mu'_R) \ni 
(drip(\sigma),\mu_{PQ},*,*,*)$.

When considering MBD actions and, in particular, the
mate action, 
we have to do with another kind of causality called {\it  environmental} in \cite{busi}, due to the fact that the interaction possibilities of
the child membranes are increased by the mate synchronisation.

Examples of this kind of causality can be observed in our running examples.
In the first, for instance, 
the $bud_o$  depends on the $mate_n$, as reflected by the CFA entries:
${\cal C}(\mu_{R2}) \ni (bud_o, \mu_P,bud^{\bot}_o,\mu_Q,*,*,\mu_{PQ})$, where
${\cal C}(\mu_{PQ}) \ni (mate_n, \mu_P,mate^{\bot}_n,\mu_Q,*,*,*)$.

In the second, we can observe that
the synchronisation on $mate_m$ cannot be performed before a synchronisation on $mate_n$, as captured by the following CFA entries:
${\cal C}(\mu_{P_0Q_0}) \ni (mate_m, \mu_{P_0},mate^{\bot}_m,\mu_{Q_0},*,*,\mu_{PQ})$, and
${\cal C}(\mu_{P_0P_1Q_0}) \ni$ $(mate_m, \mu_{P_0P_1},mate^{\bot}_m,\mu_{Q_0},*,*,\mu_{PQ})$, where $(mate_n, \mu_P,mate^{\bot}_n,\mu_Q,*,*,*)$ belongs to
${\cal C}(\mu_{PQ})$.

Finally, in \cite{busi}, a casual dependency generated by Bud (and Drip) is discussed on the following example:

$
\begin{array}{l}
P = bud_n^{\bot} (drip(\sigma)) \langle bud_n \langle \rangle^{\mu_{P0}} \rangle^{\mu_P}
\rightarrow_{bud} 
drip(\sigma) \langle  \langle \rangle^{\mu_{P0}} \rangle^{\mu_R} \circ \langle \rangle^{\mu_{P}}
\rightarrow_{drip} \\
 \sigma \langle \rangle^{\mu_{R_{\sigma}}} \circ \langle  \langle \rangle^{\mu_{P0}} \rangle^{\mu_R} \circ \langle \rangle^{\mu_{P}}
\end{array}
$

\noindent
The bud action generates a new membrane and the corresponding actions are caused by the new membrane, as captured by the CFA entries:
$drip(\sigma) \in {\cal I}(*,*,\mu_{R})$ and 
${\cal C}(\mu_{R_{\sigma}}) \ni (drip(\sigma), \mu_R,*,*,*)$.

These considerations encourage us to further investigate and to formalise the static contribution of the CFA in establishing causal relationships.

\section{The Analysis at Work: Viral Infection}
\label{viral}

We illustrate our approach by applying it to the abstract description of the 
infection cycle of the Semliki Forest Virus, shown in Figure~\ref{virus},
as specified in  \cite{C05}.
The Semliki Forest Virus is one of 
the so-called ``enveloped viruses''. 
We focus just on the first stage of the cycle and we report the analysis as given in~\cite{BBC09}.
The virus, specified in Table~\ref{virus-evol}, consists of a capsid containing the viral RNA (the nucleocapsid).
The nucleocapsid is surrounded by a membrane, similar to the cellular one, but enriched with a special protein.
The virus is brought into the cell by phagocytosis, thus wrapped by an additional membrane layer.
An endosome compartment is merged with the wrapped-up virus.
At this point, the virus uses its special membrane protein to trigger the exocytosis process that leads the naked nucleocapsid into the cytosol, ready to damage it.
By summarising, if the $cell$ gets close to a $virus$, then it evolves into an infected cell.
%
\begin{figure*}[htb]\centering
\includegraphics[width=0.8\textwidth]{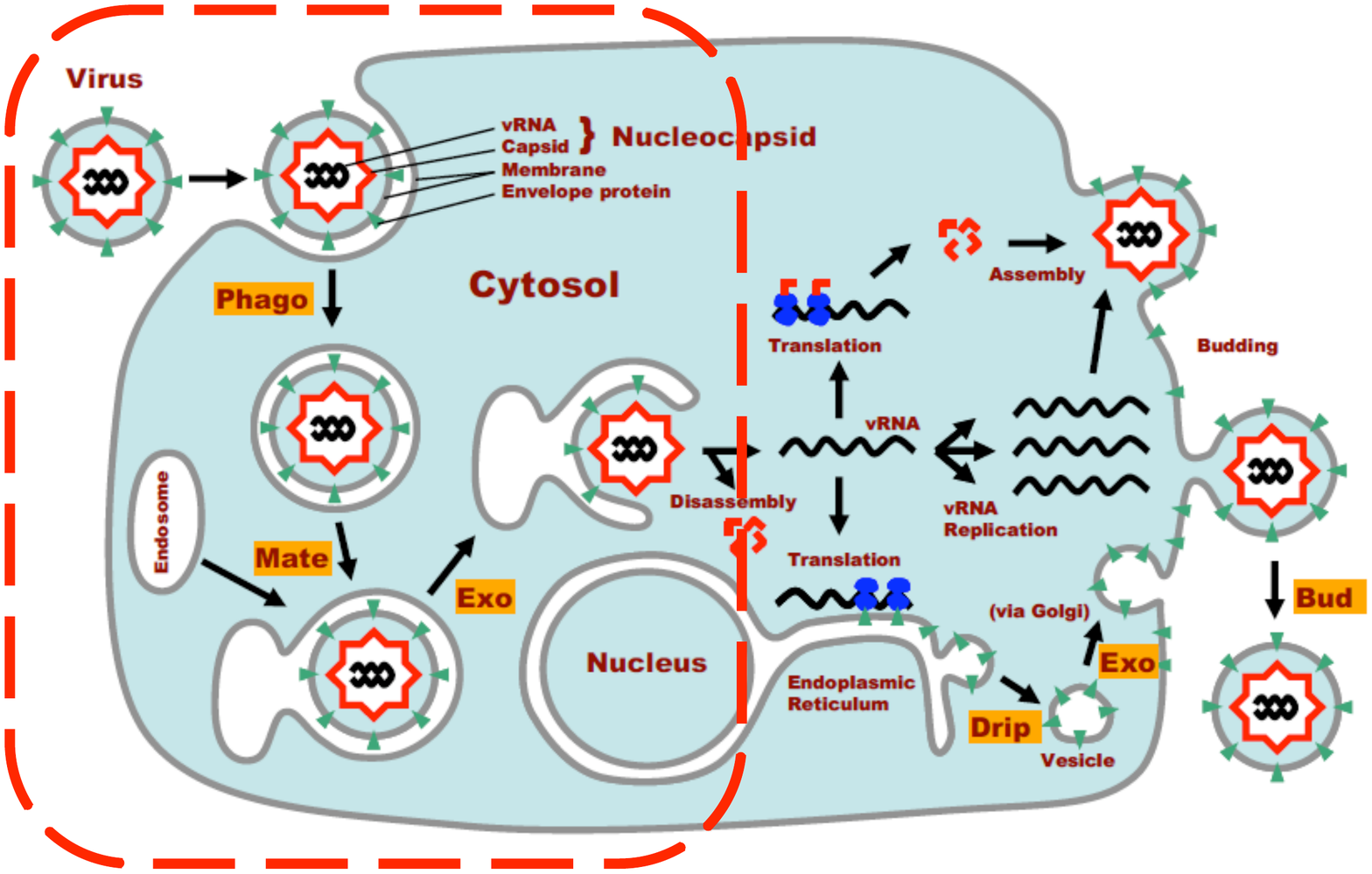}
\caption{\it  Viral Infection (highlighted part) and Reproduction. [Adapted from \cite{C05} and \cite{alberts}]}
\label{virus}
\end{figure*}
%
%
\noindent
The complete evolution of the viral infection 
is reported in 
Table~\ref{virus-evol}, while 
the main analysis entries are 
in 
Table~\ref{virus-cfa}. 
The specification includes the PEP version of Brane calculus, whose syntax and reduction semantics is reported in Table~\ref{opsem2}.
This further set of PEP actions ($\Xi_{PEP}$) are inspired by endocytosis and exocytosis processes.
The first indicates the process of incorporating external material into a cell, by engulfing it with the cell membrane, while the second one indicates the reverse process.
Endocytosis is rendered by two more basic operations:  {\it  phagocytosis} (denoted by $phago$), that consists in engulfing just one external membrane, and {\it  pinocytosis}  (denoted by $pino$), consists in engulfing zero external membranes;  {\it  exocytosis} is instead denoted by $exo$.
The CFA for the calculus can be straightforwardly extended to deal with the Phago/Exo/Pino (PEP)
actions, as shown in Table~\ref{analysis2}.

\begin{table}[th]
{\small
$$\begin{array}{|c|}
\hline
\begin{array}{lll}
a ::= &
phago_n \ | \ phago_n^{\bot}(\rho)  \ | \ 
exo_n \ | \ exo_n^{\bot} \ | \ 
pino(\rho)
& \mbox{\ \  $\Xi_{{PEP}}$}
\end{array}
\\
\hline
\begin{array}{ll}
(Phago) &
\textcolor{blue}{phago_n}.\sigma | \sigma_0 \langle P \rangle^{\mu_P} \circ
\textcolor{blue}{phago_n^{\bot}(\rho)}.\tau | \tau_0 \langle Q \rangle^{\mu_Q} 
\rightarrow
\tau | \tau_0 \langle \rho 
\langle
\sigma | \sigma_0 . \langle P  \rangle^{\mu_{P}} \rangle^{\mu_R} \circ Q \rangle^{\mu_{Q}} \\
& \mbox{ where }
\mu_{R} = {\bf MI_{phago}}(phago_n,\mu_P,phago_n^{\bot}(\rho),\mu_Q,\mu_{gp},\mu_p,\mu) 
\\
&
\mbox{ and $\mu$ identifies the closest membrane surrounding $\mu_P$ in the context $\mu_{gp}\mu_p$}
\\
(Exo) &
\textcolor{blue}{exo_n^{\bot}}.\tau | \tau_0 \langle \textcolor{blue}{exo_n}.\sigma | \sigma_0 
\langle P  \rangle^{\mu_{P}} \circ Q  \rangle^{\mu_{Q}} 
\rightarrow
P \circ 
\sigma | \sigma_0 | \tau | \tau_0 \langle Q \rangle^{\mu_Q}
\\
(Pino) &
\textcolor{blue}{pino(\rho)}.\sigma | \sigma_0 \langle P  \rangle^{\mu_{P}} 
\rightarrow
\sigma | \sigma_0 \langle \rho \langle \rangle^{\mu_R} \circ P \rangle^{\mu_P}
\\
& \mbox{ where }
\mu_{R} = {\bf MI_{pino}}(pino(\rho),\mu_P,\mu_{gp},\mu_p,\mu)
\\
&
\mbox{ and $\mu$ identifies the closest membrane surrounding $\mu_P$ in the context $\mu_{gp}\mu_p$}
\end{array}
\\
[4ex]
\hline
\hline
\end{array}$$
\caption{Syntax and Reduction Rules for PEP Actions.}
\label{opsem2}
}
\end{table}
\begin{table}
\begin{center}
{
\begin{tabular}{|c|} \hline
\mbox{$\begin{array}{ll}
(Phago) & \ phago_n \in {\cal I}(\mu_{p}, \mu, \mu_{P})
\wedge
phago_n^{\bot}(\rho) \in {\cal I}(\mu_{p}, \mu, {\mu_{Q}}) \wedge 
\mu_{P},\mu_{Q} \in {\cal I}(\mu_{gp},\mu_{p},\mu) \ \wedge \\

& \ ((\mu_{p}, \mu, {\mu_{P}}),(\mu_{p}, \mu, {\mu_{Q}})) \not\in {\cal R}
\\
&
\Rightarrow 
A(\rho) \subseteq {\cal I}(\mu, \mu_Q,\mu_R) \wedge \mu_{R} \in {\cal I}(\mu_p,\mu,\mu_Q) \wedge \mu_P \in {\cal I}(\mu, \mu_Q,\mu_R) 
\\
& \mbox{ where }
\mu_{R} = {\bf MI_{phago}}(phago_n,\mu_P,phago_n^{\bot}(\rho), \mu_Q,\mu_{gp},\mu_{p},\mu)
\\
(Exo) & \ exo_n \in {\cal I}(\mu, \mu_Q, \mu_{P})
\wedge
exo_n^{\bot} \in {\cal I}(\mu_{p}, \mu, {\mu_{Q}}) \wedge \mu_{P} \in {\cal I}(\mu_{p}, \mu, \mu_{Q}) \wedge \mu_{Q} \in {\cal I}(\mu_{gp},\mu_{p},\mu)
\\
&
\Rightarrow
A(\sigma), A(\sigma_0) \subseteq {\cal I}(\mu_{p}, \mu,\mu_Q)
\wedge {\cal I}(\mu, \mu_Q, \mu_{P}) \subseteq {\cal I}(\mu_{gp},\mu_{p},\mu)
\\
(Pino) & \ pino(\rho) \in {\cal I}(\mu_{p}, \mu,\mu_{P})
\wedge
\mu_{P} \in {\cal I}(\mu_{gp},\mu_{p},\mu)
\\
&
\Rightarrow 
A(\rho) \subseteq {\cal I}(\mu, \mu_P, \mu_R) \wedge
\mu_{R} \in {\cal I}(\mu_p,\mu,\mu_P) 
\\
& \mbox{ where }
\mu_{R} = {\bf MI_{pino}}(pino(\rho),\mu_P,\mu_{gp},\mu_{p},\mu)
\\[0.1ex]
\end{array} $} \\ \hline
\end{tabular}
}
\end{center}
\caption{Closure Rules for PEP Actions}
\label{analysis2}
\end{table}%
Roughly, the analysis results allow us to predict the effects of the infection.
Indeed, the inclusion $\mu_{nucap} \in {\cal I}(*,*,\mu_{memb})$ reflects the fact that, at the end of the shown computation,
$nucap$ is inside $membrane$ together with $cytosol'$ that is equivalent to $cytosol$, apart from the label $\mu_{ph\mbox{-}endo}$ that decorates the enclosed membrane $endosome$. 
%
Furthermore, we can check our properties in this systems.
As far as the spatial structure properties, we can prove here, e.g., that
(i) the capability
$exo^{\bot}$ never affects the membrane labelled $\mu_{ph}$ 
(as $exo^{\bot} \not\in {\cal I}(*,\mu_{memb},\mu_{ph})$); and that
(ii) the membrane 
$\mu_{virus}$ 
never ends up inside the membrane labelled $\mu_{endo}$ 
(as $\mu_{virus} \not\in {\cal I}(*,\mu_{memb},\mu_{endo})$).
Furthermore, we can observe that the CFA captures the dependency of the synchronisation on $mate$
on the synchronisation on $phago$, since we have that
$(mate,\mu_{ph},mate^{\bot},\mu_{endo},*,*,\mu_{memb}) \in{\cal C}(\mu_{ph\mbox{-}endo})$, and
$\mu_{ph}$ is such that we have that $(phago,\mu_{virus}, phago^{\bot},\mu_{memb},*,*,*)\in{\cal C}(\mu_{ph})$.

%
%

\begin{table}
\begin{center}
{
\small
\begin{tabular}{|c|} \hline
\mbox{$
\begin{array}{l}
\begin{array}{lll}
virus & \stackrel{def}{=}  & phago.exo \langle nucap \rangle^{\mu_{virus}} \\
nucap & \stackrel{def}{=}  &  !bud | X \langle vRNA \rangle^{\mu_{nucap}}  \\
cell & \stackrel{def}{=}  & membrane \langle  cytosol  \rangle^{\mu_{memb}}  \\
\end{array} 
\ \ \ \ \
\begin{array}{lll}
membrane & \stackrel{def}{=} & ! phago^{\bot}(mate) | !exo^{\bot} \\
cytosol & \stackrel{def}{=}  & endosome \circ Z \\
endosome & \stackrel{def}{=}  & !mate^{\bot} | ! exo^{\bot} \langle \rangle^{\mu_{endo}} 
\\
\end{array} \\
\begin{array}{c}
\hline
virus \circ cell  \\
\equiv (\textcolor{blue}{phago}.exo) \langle nucap \rangle^{\mu_{virus}}  \circ 
(! \textcolor{blue}{phago^{\bot}(mate)} | !exo^{\bot}) \langle  cytosol  \rangle^{\mu_{memb}}  \inter{phago} \\
(! phago^{\bot}(mate) | !exo^{\bot}) \langle \textcolor{blue}{mate} \langle exo \langle nucap \rangle^{\mu_{virus}} \rangle^{\mu_{ph}} \circ 
(! \textcolor{blue}{mate^{\bot}} | ! exo^{\bot}) \langle \rangle^{\mu_{endo}} \circ Z \rangle^{\mu_{memb}}
 \\
\inter{mate} (! phago^{\bot}(mate) | !exo^{\bot}) \langle (! mate^{\bot} | ! \textcolor{blue}{exo^{\bot}}) \langle \textcolor{blue}{exo} 
\langle nucap \rangle^{\mu_{virus}} \rangle^{\mu_{ph\mbox{-}endo}} \circ Z
\rangle^{\mu_{memb}}  \inter{exo} \\
(! phago^{\bot}(mate) | !exo^{\bot}) \langle (! mate^{\bot} | ! exo^{\bot})\langle  \rangle^{\mu_{ph\mbox{-}endo}} 
\circ nucap \circ Z \rangle^{\mu_{memb}} \equiv \\
membrane \langle  nucap \circ cytosol'  \rangle^{\mu_{memb}}\\
\end{array} 
\\[2ex]
\hline

\end{array} $} \\ \hline
\end{tabular}
}
\end{center}
\caption{Viral Infection System and its Evolution}
\label{virus-evol}
\end{table}
\begin{table}
\begin{center}
{
\begin{tabular}{|c|}
 \hline
\mbox{$
\begin{array}{l}
\begin{array}{l}
\mu_{nucap} \in {\cal I}(*,*,\mu_{virus}),\mu_{endo} \in {\cal I}(*,*,\mu_{memb}),
\mu_{virus}, \mu_{memb} \in {\cal I}(*,*,*)   \\
phago, exo \in {\cal I}(*,*,\mu_{virus}),
phago^{\bot}(mate),exo^{\bot} \in {\cal I}(*,*,\mu_{memb})
\\[1ex]
mate \in {\cal I}(*,\mu_{memb},\mu_{ph}),  \mu_{ph} \in {\cal I}(*,*,\mu_{memb}), 
\mu_{virus} \in {\cal I}(*,\mu_{memb},\mu_{ph}),\\
mate^{\bot}, exo^{\bot} \in {\cal I}(*,\mu_{memb},\mu_{endo})\\
\mu_{nucap} \in {\cal I}(\mu_{memb},\mu_{ph},\mu_{virus}),\\
 \mu_{ph\mbox{-}endo} \in {\cal I}(*,*,\mu_{memb}), 
\mu_{virus} \in {\cal I}(*,\mu_{memb},\mu_{ph\mbox{-}endo}),\\
mate^{\bot}, exo^{\bot} \in {\cal I}(*,\mu_{memb},\mu_{ph\mbox{-}endo})\\
\mu_{nucap} \in {\cal I}(\mu_{memb},\mu_{ph\mbox{-}endo},\mu_{virus}),\\
\mu_{nucap} \in {\cal I}(*,*,\mu_{memb}),\\
\mu_{ph-endo} \in {\cal I}(*,*,\mu_{memb})\\
 (phago,\mu_{virus}, phago^{\bot},\mu_{memb},*,*,*)\in{\cal C}(\mu_{ph}) \\
 (mate,\mu_{ph},mate^{\bot},\mu_{endo},*,*,\mu_{memb}) \in{\cal C}(\mu_{ph\mbox{-}endo}) 
\\
\end{array}
\\[2ex]
\end{array} $} \\ \hline
\end{tabular}
}
\end{center}
\caption{Viral Infection Analysis Results}
\label{virus-cfa}
\end{table}

\section{Conclusions}
\label{concl}

We have presented a refinement of the CFA
for the Brane calculi \cite{BBC09}, based on contextual and causal information. 
The CFA provides us with a verification framework for properties of biological systems modelled in Brane,
such as properties on the spatial structure of processes, in terms of membrane hierarchy.
We plan to formalise new properties like the ones introduced here.

We have found that the CFA is able to capture some kinds of causal dependencies \cite{busi} arising in the MBD version of 
Brane Calculi.
As future work, we would like to investigate thoroughly and formally the static contribution of the CFA in establishing causal relationships between the Brane interactions.

\smallskip

\noindent
{\bf Acknowledgments.}
We wish to thank Francesca Levi for our discussion on a draft of our paper and our anonymous referees for their useful comments.


\bibliographystyle{mecbic} 

\appendix

\section{Proofs}\label{app-proof}

This appendix restates the lemmata and theorems presented earlier
in the paper and gives the proofs of their correctness.
To establish the semantic correctness, the following auxiliary results are needed. 

\begin{proposition}\label{Afatto}
If 
$\FORMb{\mu_p \mu \mu_1}{P}$
and
${\cal I}(\mu_p, \mu, \mu_1) \subseteq {\cal I}(\mu_p, \mu, \mu_2)$,
then 
$\FORMb{\mu_p \mu \mu_2}{P}$.
\end{proposition}
\begin{proof}
By structural induction on $P$. 
We show just one case.
\\
\textbf{Case} $P = \sigma \langle P' \rangle^{\mu_s}$. 
We have that \FORMb{\mu_p \mu \mu_1}{P} is equivalent to
$\mu_s \in {\cal I}(\mu_p, \mu, \mu_1) \ \wedge \ 
A(\sigma) \subseteq {\cal I}(\mu, \mu_1, \mu_s)
\ \wedge \
\FORMb{\mu \mu_1 \mu_s}{P'}$.
Now, $\mu_s \in {\cal I}(\mu_p, \mu, \mu_1)$ and ${\cal I}(\mu_p, \mu, \mu_1) \subseteq {\cal I}(\mu_p, \mu, \mu_2)$ 
and $A(\sigma) \subseteq {\cal I}(\mu, \mu_1, \mu_s)$
imply
$\mu_s \in {\cal I}(\mu_p, \mu, \mu_2)$ and
$A(\sigma) \subseteq {\cal I}(\mu, \mu_2, \mu_s)$.
Therefore,
by induction hypothesis, we have that 
$\FORMb{\mu_p \mu \mu_2}{P}$. 
\end{proof}

\begin{proposition}\label{AfunctA}
If $\sigma \equiv \tau$ then $A(\sigma) = A(\tau)$.
\end{proposition}
\begin{proof}
The proof amounts to a straightforward inspection
of each of the
clauses defining the structural congruence clauses relative to membranes.
We only show two cases, the others are similar.
 \\
 \textbf{Case} $\sigma_0 | \sigma_1 \equiv \sigma_1 | \sigma_0$. 
 We have that
 $A(\sigma_0 | \sigma_1) = A(\sigma_0) \cup A(\sigma_1) = 
 A(\sigma_1 | \sigma_0)$.
 \\
 \textbf{Case} $\sigma \equiv \tau \Rightarrow \sigma | \rho \equiv \tau | \rho$.  
 We have that 
$A(\sigma | \rho) = A(\sigma) \cup A(\rho)$.
Now, since $\sigma \equiv \tau$, we have that $A(\sigma)=A(\tau)$ and therefore
$A(\sigma | \rho) = A(\tau) \cup A(\rho)$, from which the required $A(\tau | \rho)$.
\end{proof}

\noindent
{\bf Lemma 4.1} \textbf{(Invariance of Structural Congruence)}\label{Acongr}
{\it  If $P \equiv Q$ and we have that  $\FORMb{\mu_{gp} \mu_p \mu}{P}$  then also 
$\FORMb{\mu_{gp} \mu_p \mu}{Q}$}.
\begin{proof}
The proof amounts to a straightforward inspection
of each of the
clauses defining the structural congruence clauses. 
We only show two cases, the others are similar.
 \\
  \textbf{Case} $P_0 \circ P_1 \equiv P_1 \circ P_0$.  
 We have that 
 $ \FORMb{\mu_{gp} \mu_p \mu}{P_0 \circ P_1}$ is equivalent to 
 $\FORMb{\mu_{gp} \mu_p \mu}{P_0}  \wedge \FORMb{\mu_{gp} \mu_p \mu}{P_1}$, that is equivalent to
 $\FORMb{\mu_{gp} \mu_p \mu}{P_1} \wedge \FORMb{\mu_{gp} \mu_p \mu}{P_0}$ and therefore to
 $\FORMb{\mu_{gp} \mu_p \mu}{P_1 \circ P_0}$.
\\
 \textbf{Case} $P \equiv Q \ \wedge \ \sigma \equiv \tau  
 \Rightarrow \sigma \langle P \rangle^{\mu_s} \equiv \tau \langle Q \rangle^{\mu_s}$.
 We have  that 
 $\FORMb{\mu_{gp} \mu_p \mu}{\sigma \langle P \rangle^{\mu_s}} $ is equivalent to 
 $\mu_s \in {\cal I}(\mu_{gp},\mu_p, \mu) \ \wedge \ 
A(\sigma) \subseteq {\cal I}(\mu_p, \mu,\mu_s)
\ \wedge \
\FORMb{\mu_p, \mu,\mu_s}{P}$.
 By Proposition~\ref{AfunctA}, $A(\tau) \subseteq {\cal I}(\mu_p, \mu,\mu_s)$, and by
 induction hypothesis, we have that 
 $\FORMb{\mu_p, \mu,\mu_s}{Q}$. As a consequence, we can conclude that
 $\FORMb{\mu_{gp} \mu_p \mu}{\tau \langle Q \rangle^{\mu_s}} $.
 \end{proof}

\noindent
{\bf Theorem 4.2} \textbf{(Subject Reduction)}\label{Asbj-red}\\
{\it  If $P \rightarrow Q$ and $\FORMb{\mu_{gp} \mu_p \mu}{P}$ then also 
$\FORMb{\mu_{gp} \mu_p \mu}{Q}$}.

\begin{proof}
The proof is by induction on $P \rightarrow Q$.
The proofs for the rules $(Par)$ and $(Brane)$ are straightforward, using the induction hypothesis and the clauses in Table~\ref{analysis}.
The proof for the $(Struct)$ uses instead the induction hypothesis and Lemma 4.1.
The proofs for the basic actions in the lower part of Table~\ref{opsem} 
are straightforward, using the clauses in Table~\ref{analysis}.
\\
\textbf{Case} (Par). 
Let $P$ be $P_0 \circ P_1$ and $Q$ be $P'_0 \circ P_1$, with $P_0 \rightarrow P'_0$. 
We have to prove that $ \FORMb{\mu_{gp} \mu_p \mu}{Q}$.
Now $ \FORMb{\mu_{gp} \mu_p \mu}{P}$ is equivalent to 
$\FORMb{\mu_{gp} \mu_p \mu}{P_0} \ \wedge \  \FORMb{\mu_{gp} \mu_p \mu}{P_1}$. 
By induction hypothesis, we have that $ \FORMb{\mu_{gp} \mu_p \mu}{P'_0}$, and from
$ \FORMb{\mu_{gp} \mu_p \mu}{P'_0} \ \wedge \  \FORMb{\mu_{gp} \mu_p \mu}{P_1}$ 
we obtain the required $ \FORMb{\mu_{gp} \mu_p \mu}{Q}$.
\\
\textbf{Case} (Brane). 
Let $P$ be $\sigma \langle P_0 \rangle^{\mu_s}$ and $Q$ be $\sigma \langle P'_0 \rangle^{\mu_s}$. 
We have to prove that $\FORMb{\mu_{gp} \mu_p \mu}{\sigma \langle P'_0 \rangle^{\mu_s}}$.
Now $\FORMb{\mu_{gp} \mu_p \mu}{P}$ is equivalent to have that
$\mu_s \in {\cal I}(\mu_{gp},\mu_p, \mu) \ \wedge \ 
A(\sigma) \subseteq {\cal I}(\mu_p, \mu,\mu_s)
\ \wedge \
\FORMb{\mu_p, \mu,\mu_s}{P_0}$. 
By induction hypothesis, we have that  $\FORMb{\mu_p, \mu,\mu_s}{P'_0}$. We can therefore conclude that
$\FORMb{\mu_{gp} \mu_p \mu}{Q}$.
\\
\textbf{Case} (Struct). 
Let $P \equiv P_0$, with $P_0 \rightarrow P_1$ such that $P_1 \equiv Q$.
By Lemma~\ref{Acongr}, we have that 
$\FORMb{\mu_{gp} \mu_p \mu}{P_0}$, by induction hypothesis $\FORMb{\mu_{gp} \mu_p \mu}{P_1}$ and, again by
Lemma~\ref{Acongr},
$\FORMb{\mu_{gp} \mu_p \mu}{Q}$.
\\
\textbf{Case} (Mate). 
Let $P$ be $mate_n.\sigma | \sigma_0 \langle P_0 \rangle^{\mu_{0}} \circ
mate_n^{\bot}.\tau | \tau_0 \langle P_1 \rangle^{\mu_1}$ and
$Q$ be $\sigma | \sigma_0 | \tau | \tau_0 \langle P_0 \circ P_1 \rangle^{\mu_{01}}$.
Then, $\FORMb{\mu_{gp} \mu_p \mu}{P}$ amounts to
$\FORMb{\mu_{gp} \mu_p \mu}{mate_n.\sigma | \sigma_0 \langle P_0 \rangle^{\mu_{0}}}$ and
$\FORMb{\mu_{gp} \mu_p \mu}{mate_n^{\bot}.\tau | \tau_0 \langle P_1 \rangle^{\mu_1}}$ and, in turn, to
$\mu_0,\mu_1 \in {\cal I}(\mu_{gp}, \mu_p, \mu)$, 
$\{mate_n\} \cup A(\sigma) \cup A(\sigma_0) \subseteq {\cal I}(\mu_p, \mu, \mu_{0})$,
$\{mate_n^{\bot}\} \cup A(\tau) \cup A(\tau_0) \subseteq {\cal I}(\mu_p, \mu, \mu_{1})$, and
$\FORMb{\mu_p \mu \mu_0}{P_0}$ and $\FORMb{\mu_p \mu \mu_1}{P_1}$.
Note that, $ ((\mu_p,\mu,\mu_{0}),(\mu_p,\mu,\mu_{1}))$ does not belong to $\mathcal{R}$.
Because of the closure conditions, from the above, we have, amongst the several implied conditions, that
$\exists \mu_{01} = {\bf MI_{mate}}(mate_n,\mu_0,mate_n^{\bot},\mu_1,\mu_{gp}, \mu_p, \mu)$ such that
$\mu_{01} \in{\cal I}(\mu_{gp},\mu_p,\mu)  \ \wedge \ 
{\cal I}(\mu_p, \mu, \mu_0) \cup {\cal I}(\mu_p, \mu, \mu_1) \subseteq {\cal I}(\mu_p, \mu, \mu_{01})$. 
From
${\cal I}(\mu_p, \mu, \mu_i)\subseteq {\cal I}(\mu_p, \mu, \mu_{01})$ for $i=0,1$, 
we have that
$A(\sigma) \cup A(\sigma_0) \subseteq {\cal I}(\mu_p, \mu, \mu_{01})$
and 
$A(\tau) \cup A(\tau_0) \subseteq (\mu_p, \mu, \mu_{01})$ and, by Proposition~\ref{Afatto}, 
we have that
$\FORMb{\mu_p \mu \mu_{01}}{P_0}$ and
$\FORMb{\mu_p \mu \mu_{01}}{P_1}$, and hence the required
$\FORMb{\mu_p \mu \mu_{01}}{Q}$.
\\
\textbf{Case} (Bud). $\!\!$ Let $P$ be $bud_n^{\bot}(\rho).\tau | \tau_0 \langle 
bud_n.\sigma | \sigma_0 \langle P_0 \rangle^{\mu_0} \circ P_1
\rangle^{\mu_1}$ and
$Q$ be the process $\rho  \langle \sigma | \sigma_0 \langle P_0 \rangle^{\mu_0} \rangle^{\mu_R} \circ \tau | \tau_0\langle P_1 \rangle^{\mu_1}$.
Now, $\FORMb{\mu_{gp} \mu_p \mu}{P}$ is equivalent to
$\mu_{1} \in {\cal I}(\mu_{gp}, \mu_p, \mu)$,
$\{bud_n^{\bot}(\rho)\} \cup A(\tau) \cup A(\tau_0) \subseteq {\cal I}(\mu_p, \mu, \mu_1)$, and, moreover,
$\FORMb{\mu_p \mu \mu_1}{bud_n.\sigma | \sigma_0 \langle P_0 \rangle^{\mu_0} }$ and
$\FORMb{\mu_p \mu \mu_1}{P_1}$, from which we have that
$\mu_{0} \in {\cal I}(\mu_p, \mu, \mu_1)$, 
$\{bud_n\} \cup A(\sigma) \cup A(\sigma_0) \subseteq {\cal I}(\mu, \mu_1,\mu_0)$, and
$\FORMb{\mu \mu_1\mu_0}{P_0}$.
Because of the closure conditions, from above, we have that
 $\exists \mu_{R} = {\bf MI_{bud}}(bud_n,\mu_0,bud_n^{\bot},\mu_1,\mu_{gp}, \mu_p, \mu)$ such that
$\mu_R \in {\cal I}(\mu_{gp}, \mu_p, \mu)$, $A(\rho) \subseteq {\cal I}(\mu_p, \mu, \mu_R)$, 
$\mu_0 \in {\cal I}(\mu_p,\mu, \mu_R)$,
and 
(${\cal I}(\mu,\mu_1,\mu_0) \subseteq {\cal I}(\mu,\mu_R,\mu_0)$ and
${\cal I}(\mu_1,\mu_0) \subseteq {\cal I}(\mu_R,\mu_0)$ (cond 2)).
We have that  $\FORMb{\mu_{gp} \mu_p \mu}{Q}$ is equivalent to have that
$\FORMb{\mu_{gp} \mu_p \mu}{\rho  \langle \sigma | \sigma_0 \langle P_0 \rangle^{\mu_0} \rangle^{\mu_R}}$ (1)
and that 
$\FORMb{\mu_{gp} \mu_p \mu}{\tau|\tau_0\langle P_1\rangle^{\mu_1}}$ (2).
For (1), we
have to prove that $\mu_R \in {\cal I}(\mu_{gp}, \mu_p, \mu)$,
$A(\rho) \subseteq {\cal I}(\mu_p, \mu, \mu_R)$ and 
$\FORMb{\mu_p \mu \mu_R}{\sigma | \sigma_0 \langle P_0 \rangle^{\mu_0} }$, that is equivalent to
$\mu_0 \in {\cal I}(\mu_p,\mu, \mu_R)$,
$A(\sigma) \cup A(\sigma_0) \subseteq {\cal I}(\mu, \mu_R,\mu_0)$ and 
$\FORMb{\mu\mu_R\mu_0}{P_0}$.
From the hypotheses, we have that
$\mu_0 \in {\cal I}(\mu_p,\mu, \mu_R)$.
Since
$A(\sigma) \cup A(\sigma_0) \subseteq {\cal I}(\mu, \mu_1,\mu_0)$ and (cond 2)
we have $A(\sigma) \cup A(\sigma_0) \subseteq {\cal I}(\mu, \mu_R)$.
From $\FORMb{\mu \mu_1\mu_0}{P_0}$, because of (cond 2) and Proposition~\ref{Afatto}, we have that
$\FORMb{\mu \mu_R\mu_0}{P_0}$.
For (2), we have to prove that $\mu_1 \in {\cal I}(\mu_{gp}, \mu_p, \mu)$,
$A(\tau) \cup A(\tau_0) \subseteq {\cal I}(\mu_p, \mu, \mu_1)$ and
$\FORMb{\mu_p \mu \mu_1}{P_1}$. All these conditions
are satisfied (see above).
Therefore, we obtain the required
$\FORMb{\mu_{gp} \mu_p \mu}{Q}$.
\\
\textbf{Case} (Drip). 
Let $P$ be $drip(\rho).\sigma | \sigma_0 \langle P_0 \rangle^{\mu_0}$ and $Q$ be
$\rho \langle  \rangle^{\mu_R}  \circ  \sigma | \sigma_0 \langle P_0\rangle^{\mu_0}$.
We have that $\FORMb{\mu_{gp} \mu_p \mu}{P}$ is equivalent to
$\mu_0 \in {\cal I}(\mu_{gp},\mu_p, \mu)$, 
$\{drip(\rho)\} \cup A(\sigma) \cup A(\sigma_0) \subseteq {\cal I}(\mu_p, \mu, \mu_0)$, and
$\FORMb{\mu_p \mu \mu_0}{P_0}$.
Because of the closure conditions, from the above,
$\exists \mu_{R} = {\bf MI_{drip}}(drip(\rho),\mu_0,\mu_{gp}, \mu_p, \mu)$ such that
$\mu_R \in  {\cal I}(\mu_{gp},\mu_p, \mu)$, 
$A(\rho) \subseteq {\cal I}(\mu_p, \mu, \mu_R)$.
We have that $\FORMb{\mu_{gp} \mu_p \mu}{Q}$ is equivalent to both
$\FORMb{\mu_{gp} \mu_p \mu}{\rho \langle  \rangle^{\mu_R}}$ and
$\FORMb{\mu_{gp} \mu_p \mu}{\sigma | \sigma_0 \langle P_0\rangle^{\mu_0}}$.
The first condition is verified, because
$\mu_R \in  {\cal I}(\mu_{gp},\mu_p, \mu)$ and $A(\rho) \subseteq {\cal I}(\mu_p, \mu, \mu_R)$.
The second amounts to
$\mu_0 \in {\cal I}(\mu_{gp},\mu_p, \mu)$ and $A(\sigma) \cup A(\sigma_0) \subseteq {\cal I}(\mu_p, \mu, \mu_0)$
and it is satisfied as well.
We therefore obtain the required
$\FORMb{\mu_{gp} \mu_p \mu}{Q}$.
\end{proof}

%

\noindent
{\bf Theorem 5.2}
{\it  Given a process $P$ including a membrane labelled $\mu$, then
if $c$ never appears on the membrane 
labelled $\mu$, then 
the capability $c$ never affects the membrane labelled $\mu$}.
 
\begin{proof}
First of all, we observe that if $c$ affects $\mu$ in $P$, then we have a contradiction, since it implies that
$c \not\in {\cal I}(\mu_{gp},\mu_p,\mu)$ for some context $\mu_{gp}\mu_p$.
We now show that there exist no $Q$, $Q'$ such that $P \rightarrow^{*} Q \rightarrow Q'$ such that
$c$ does not affect $\mu$ in $Q$, while it does in $Q'$.
The only case in which this can happen is when a $(Bud)$ or a $(Drip)$ is performed with parameter $\rho$ including $c$.
Indeed, the firing of such an action lets arise a new membrane affected by the corresponding parameter.
We focus on the second one. Suppose we have in $Q$ a sub-process $drip(\rho).\sigma | \sigma_0 \langle Q_0 \rangle^{\mu}$ and that 
$c$ occurs in $\rho$. 
This amounts to have that $c$ can affect $\mu$ in $Q'$. 
By theorem~\ref{sbj-red}, we have that $({\cal I}, {\cal R},{\cal C})$ is an estimate also for $Q$.
Nevertheless this implies that $c \in {\cal I}(\mu_{gp},\mu_p,\mu)$, thus leading to a contradiction.
\end{proof}

\end{document}